 \documentclass[11pt]{article}

 \title{Stratified Random Sampling for Dependent Inputs}

 \author{Anirban Mondal\\
 Case Western Reserve University, Cleveland, OH 44106, USA\\
 \and
  Abhijit Mandal\\
 Wayne State University, Detroit, MI 48202, USA}

 \usepackage[authoryear,comma]{natbib}
\usepackage{hyperref}

\usepackage{epsfig, amsmath, amsfonts, amssymb, amsthm,multirow,algorithm,algorithmic}
\usepackage{graphicx}
\usepackage{setspace}

\usepackage{enumerate}
\usepackage{hyperref}
\hypersetup{citecolor=blue,colorlinks=true}

\theoremstyle{plain}

\newcommand{\beq}{\begin{equation}}
\newcommand{\eeq}{\end{equation}}
\newcommand{\ber}{\begin{eqnarray}}
\newcommand{\eer}{\end{eqnarray}}

\newcommand{\mf}{\mathbf}

\newcommand{\bit}{\begin{itemize}}
\newcommand{\eit}{\end{itemize}}
\newcommand{\iid}{\stackrel{\mathrm{i.i.d.}}{\sim}}
\newcommand{\ben}{\begin{enumerate}}
\newcommand{\een}{\end{enumerate}}

\numberwithin{equation}{section}

 \newtheorem{theorem}{Theorem}

\begin{document}

%
\maketitle
%


%
%

\begin{abstract} 

 A new approach of obtaining stratified random samples from statistically dependent random variables is described. The proposed method can be used to obtain samples from the input space of a computer forward model in estimating expectations of functions of the corresponding output variables. The advantage of the proposed method over the existing methods is that it preserves the exact form of the joint distribution on the input variables. The asymptotic distribution of the new estimator is derived.  Asymptotically, the variance of the estimator using the proposed method is less than that obtained using the simple random sampling, with the degree of variance reduction depending on the degree of additivity in the function being integrated. This technique is applied to a practical example related to the performance of the river flood inundation model.

\end{abstract}





\section{Introduction}
Mathematical models are widely used by engineers and scientists to describe physical, economic, and social processes. Often these models are complex in nature and are described by a system of ordinary or partial differential equations, which cannot be solved analytically. Given the values of the input parameters of the processes, complex computer codes are widely used to solve such systems numerically, providing the corresponding outputs. These computer models, also known as forward models, are used for prediction, uncertainty analysis, sensitivity analysis, and model calibration. For such analysis, the computer code is evaluated multiple times to estimate some function of the outputs. However, complex computer codes are often too expensive to be directly used to perform such studies. Design of computer experiment plays an important role in such situations, where the goal is to choose the input configuration in an optimal way such that the uncertainty and sensitivity analysis can be done with a few number of simulation runs. Moreover, to avoid the computation challenge, it is often useful to replace the computer model either by a simpler mathematical approximation called as surrogate or meta model (\citealp{volkova, simpson1}), or by a statistical approximation, called emulator (\citealp{oakley2002,ohagan2006,conti2009}). Such surrogate models or emulators are also widely used for the model calibration \citep{ken2001}, which is the inverse problem of inferring about the model input parameters given the outputs. The surrogate or the emulator model is built based on set of training simulations of the original computer model. The full optimal exploration  of the variation domain of the input variables is therefore very important in order to avoid non-informative simulation points \citep{fang,sacks1989}.

In most situations, the simple random sampling (SRS) needs very large number of samples to achieve a fixed level of efficiency in estimating some functions of the outputs of such complex models. The required sample size increases rapidly with an increase in dimension of the input space. On the other hand, the Latin hypercube sampling (LHS) introduced by \cite{mckay1} needs smaller sample sizes than the SRS to achieve the same level of efficiency. Compared to the SRS, which  only ensures independence between samples, the LHS also ensures the full coverage of the range of the input variables through stratification over the marginal probability distributions of the inputs. 
Thus, the LHS method is more efficient and robust if the components of the inputs are independently distributed. In case of dependent inputs \cite{iman} has proposed an approximate version of the LHS based on the rank correlation. \cite{stein1} further improved this procedure such that each sample vector has approximately the correct joint distribution when the sample size is large. However, it is noticed that in many scenarios, the rank-based method results in a large bias and small efficiency in the estimators, particularly for small sample sizes. Moreover, the  joint distribution of the inputs are also not completely preserved in these sampling schemes even for moderately large sample sizes. Therefore, these rank-based techniques are not very useful for many real applications where one is restricted to a small sample size due to the computational burden of the expensive forward simulator. To overcome this situation, we propose a novel sampling scheme for dependent inputs that precisely gives a random sample from the target joint distribution while keeping the mean squared error of the estimator smaller than the existing methods. In the traditional LHS, where the components of inputs are independent, the stratification of the marginal distributions leads to the stratification of the joint distribution. However, for dependent random variables, all the marginal distributions and the joint distribution cannot be stratified simultaneously. Here, we propose a new sampling scheme, called the Latin hypercube sampling for dependent random variables (LHSD), where we ensure that the conditional probability distributions of the inputs are stratified. The main algorithm of the LHSD is similar to the traditional LHS, hence it retains all important properties of LHS. The joint distribution of the inputs is preserved in our sampling scheme as it is precisely the product of the conditional distributions. 

In some practical situations the joint probability distribution of the inputs, and hence the corresponding conditional probability distributions may be unknown. In these situations, we propose  a copula-based method to construct the joint probability distribution from the marginal distributions. For finite sample, it is shown that the variance of the estimators based on LHSD is always smaller than the variance of the estimator based on the SRS. The large sample properties of the LHSD based estimators are also provided. We consider two simulation-based examples and one practical example, where the  results show that the traditional LHS and rank-based LHS method have considerable bias in the estimator when the inputs are dependent. Our proposed LHSD outperforms the SRS, the traditional LHS, and rank-based LHS in terms of the mean squared error (MSE), for small and moderately large sample sizes. The simulation results also show that the rank-based LHS fails to retain the joint probability distribution of the input variables even for moderately large sample sizes, which is the reason for the considerable bias in the estimators. On the other hand, the joint distribution is completely retained in our proposed sampling scheme and hence the estimators are unbiased. 

The paper is organized as follows. In the next section, we first formulate the estimation problem, then describe the LHS algorithm for independent inputs  \citep{mckay1} and our proposed LHSD algorithm for dependent inputs. Another variant of the proposed method, which will be called the centered Latin hypercube sampling for dependent random variables (LHSD$_c$), is also described here. The use of copula, when the joint probability distribution of the inputs is not known, is also discussed in this section. The large sample properties of the estimators using the LHSD are discussed in Section \ref{sec:large_sample}. Section \ref{sec:simulation} provides the numerical results from two simulation examples, where the performance of different sampling schemes are compared. The application of the proposed sampling scheme to a real field example on a river model is presented in Section \ref{sec:river}. A concluding remark is given at the end of the paper, and the proofs of the theorems are provided in Appendix.

\section{The Sampling Techniques} \label{sec:sampling}

Consider a process (or device) which depends on an input random vector $\mf{X}=(X_1, X_2, \cdots , X_K)^T$ of fixed dimension $K$. Suppose we want to estimate the expected value of some measure of the process output, given by a function  $h(\mathbf{X})$. As described before, in many cases, $h$ is highly non-linear and/or $K$ is very large, so the probability distribution of $h(\mathbf{X})$ is not analytically tractable even though the probability distribution of $\mathbf{X}$ is completely known. For example, $h$ could depend on the output of a physical process that involves a system of partial differential equations. In such cases, Monte Carlo methods are usually used to estimate  $\tau = E(h(\mathbf{X}))$, the expected value of $h(\mathbf{X})$. Suppose $\mathbf{X}_1, \mathbf{X}_2, \cdots, \mathbf{X}_N$ form a random sample of size $N$, generated from the joint distribution of $\mathbf{X}$ using a suitable sampling scheme, then $\tau$ is estimated by $\hat{\tau}=N^{-1}\sum_{j=1}^N h(\mathbf{X}_j)$. Since, in terms of time or other complexity, it is costly to compute $h(\mathbf{X})$, we are interested in a sampling scheme that estimates $\tau$ efficiently while keeping the sample size $N$ as small as possible.

\subsection{Latin Hypercube Sampling from Independent Random Variables} \label{LHS_ind}
Suppose ${\mathbf{X}}=(X_1, X_2, \cdots, X_K)^T$ is a random vector with $K$ components. Let $F$ be the joint cumulative density function (c.d.f.) of $\mathbf{X}$, and $F_k$ be the marginal distribution function of $X_k$, $k=1, 2, \cdots K$.  \cite{mckay1} proposed a sampling scheme to generate a Latin hypercube sample of size N from $\mathbf{X}$ assuming that  the components of $\mathbf{X}$ are mutually independent. The details of the LHS algorithm is given in Algorithm \ref{lhs}. 

\begin{algorithm}
\caption{LHS}
\label{lhs}
\begin{enumerate}
\item Generate $P=((p_{jk}))_{N \times K}$, a $N \times K$ random matrix, where each column of P is an independent random permutation of ${1,2,\cdots, N}$. 

\item Using the simple random sampling method generate $NK$ independent and identically distributed (i.i.d.) uniform random variables $u_{jk}$ over $[0,1]$, i.e., $u_{jk}\iid U(0,1)$ for  $ j=1,2,\cdots,N$ and $k=1, 2, \cdots, K$. 

\item The LHS of size N from independent $X_k$'s are given by
$x_{jk}=F_k^{-1}((p_{jk}-u_{jk})/N),\ j=1,2,\cdots N,\ k=1,2,\cdots, K$, where the inverse function $F_k^{-1}(\cdot)$ is the quantile defined by $F_k^{-1}(y)=\inf\{x:F_k(x)\geq y\}$.

\end{enumerate}
\end{algorithm}

For the above  algorithm, all the marginal distributions of $\mathbf{X}$ are stratified. Such marginal stratification leads to the stratification of the joint distribution and hence there is a reduction in variance of the estimator compared to the simple random sampling.
Note that, for any fixed $j$, $(x_{j1}, x_{j2}, \cdots, x_{jK})$ form a random sample from the   marginal distribution of $X_j$, where $j=1,2, \cdots,K$. But, as a whole, they are not a set of random samples from the joint distribution of $\mathbf{X}$  unless the component variables $X_1, X_2, \cdots , X_K$ are mutually independent. Hence the estimator $\hat{\tau}$, based on samples from LHS, becomes a biased estimator for $\tau$, when  $X_1, X_2, \cdots,  X_K$ are not mutually independent. This bias could lead to a large MSE of the estimator even if there is a significant reduction of variance due to stratification.

\subsection{Latin Hypercube Sampling from Dependent Random Variables} \label{algo:LHSD}
In this section, we propose a general method for the Latin Hypercube sampling, where the components of $\mathbf{X}$ are not necessarily independent. Let us define $F_k(x_k|X_1=x_1, X_2=x_2, \cdots, X_{k-1}=x_{k-1})$ as the conditional distribution function of $X_k$ given $X_1=x_1, X_2=x_2, \cdots, X_{k-1}=x_{k-1}$ for $k= 2, 3, \cdots, K$. First, we transform $\mathbf{X}$ to ${\mathbf{Z}} = (Z_1, Z_2, \cdots, Z_K)^T$, such that $Z_1=F_1(x_1)$ and $Z_k=F_k(x_k|X_1=x_1, X_2=x_2, \cdots, X_{k-1}=x_{k-1}), \ k=2, 3, \cdots, K$. Note that the components of ${\mathbf{Z}}$ are i.i.d. $U(0,1)$. Then, we generate a LHS of size $N$ from $\mathbf{Z}$ using the LHS algorithm, and finally, we convert them back to the random sample of $\mathbf{X}$ using the inverse distribution function of $\mathbf{Z}$. The different steps for the proposed LHSD algorithm are given in Algorithm \ref{lhsd}.

\begin{algorithm}
\caption{LHSD}
\label{lhsd}
\begin{enumerate}
\item Suppose the components of ${\mathbf{Z}}$ are i.i.d. $U(0,1)$. Using  Algorithm \ref{lhs} generate $z_{jk},\ j=1,2,\cdots N, \ k=1,2.\cdots, K$, a LHS of size $N$ from ${\mathbf{Z}}$. 


\item Sequentially get ${\mathbf{x}}_j =(x_{j1}, x_{j2}, \cdots, x_{jK})^T$, the $j$-th random sample from $\mathbf{X}$ for $j=1,2, \cdots,N$,  as given below:
\begin{itemize}
\item $x_{j1} = F^{-1}_1(z_{j1})$,
\item $x_{j2} = F^{-1}_2(z_{j2}|X_1=x_1)$,
\item $x_{j3} = F^{-1}_3(z_{j3}|X_1=x_1, X_2=x_2)$ \\
 \vdots
\item $x_{jK} = F^{-1}_K(z_{jK}|X_1=x_1, X_2=x_2, \cdots, X_{k-1}=x_{k-1})$,
\end{itemize}
where  $F_k^{-1}(\cdot|X_1=x_1, X_2=x_2, \cdots, X_{k-1}=x_{k-1})$ is the inverse distribution function of $X_k$ given $X_1=x_1, X_2=x_2, \cdots, X_{k-1}=x_{k-1}$ for $k=1, 2, \cdots, K$.
\end{enumerate}
\end{algorithm}

In this sampling scheme, the conditional distributions  of $\mathbf{X}$ are stratified, but not all the marginal distributions except for the first one. In fact, it can be shown that, in case of dependent random variables, all marginal distributions cannot be stratified simultaneously while preserving the complete joint distribution. However, if the components of $\mathbf{X}$ are independently distributed, then both the marginal and conditional distributions are stratified by this method. In this case, the proposed algorithm turns out to be the traditional LHS sampling scheme proposed by \cite{mckay1}. The advantage of the proposed method over the traditional LHS is that, the joint distribution of $\mathbf{X}$ is fully preserved by this sampling scheme, even when its component are dependent, because the product of the conditional distributions is precisely the corresponding joint distribution. As a result $\hat{\tau}$ becomes an unbiased estimator of $\tau$ when the samples from LHSD are used. The stratification of the conditional distributions also guarantees reduction in variance of the estimator. For these two reasons the LHSD sampling method is more efficient in terms of MSE of the estimator, when the components of $\mathbf{X}$ are dependent.

Note that, the definition of $\mathbf{Z}$ depends on the order of the components of $\mathbf{X}$. Moreover, after generating a random sample from $\mathbf{Z}$, we sequentially recover the components of the random sample for $\mathbf{X}$. So, our method depends on the order of the components of $\mathbf{X}$ in the definition of $\mathbf{Z}$. Theoretically, we can take any order to get a random sample from $\mathbf{X}$. However, the efficiency of the sampling scheme is maximized if we arrange the components of $\mathbf{X}$ in the descending order of conditional sensitivity. It will be further discussed in the next sections.   

\subsubsection{Centered Latin Hypercube Sampling } \label{sec:lhsd_cen}
Algorithm \ref{lhs} shows that the Latin hypercube sampling has two stages of  randomization process. In the first stage, $p_{jk}$ decides in which stratum the $j$-observation of $X_k$ will belong, where $j =1, 2, \cdots, N$ and $k=1, 2, \cdots, K$. Then, in the second stage, $u_{jk}$ fixes the position of the observation in that stratum.  The variability in the estimation of $\tau$ comes from these two sources of randomization. So, it is obvious that if we fix $u_{jk}=1/2$, the small sample variance of $\hat{\tau}$ is expected to reduce  further. It  introduces a small bias as it can be easily shown that $E(\hat{\tau}) \neq \tau$. However, the simulation studies show that the reduction in the variance is  larger than the bias term, and as a result, the MSE of $\hat\tau$ is reduced in small samples.  \cite{stein1} has proved that as the sample size $N \rightarrow \infty$, there is no role of $u_{jk}$ in the asymptotic distribution of $\hat{\tau}$ of the traditional LHS given in Algorithm \ref{lhs}. Therefore, the large sample properties of the modified LHS remain unchanged. Now, fixing $u_{jk}=1/2$ implies that we are choosing the sample from the middle of the stratum, so  the modified LHS can be called as the centered LHS. Note that the strata are randomized by $p_{jk}$ in the first stage of the algorithm, so the centered LHS is also a random sampling scheme where all observations are stratified. As our primary goal of this paper is to draw a LHS from a set of dependent random variables, we will apply this method in Algorithm \ref{lhsd}, where at the first stage, we generate $z_{jk}$ using the centered LHS instead of the traditional LHS; then rest of the method will remain unchanged. We denote the modified sampling scheme as the centered LHSD or LHSD$_c$. 

\subsubsection{LHSD using Copula} \label{sec:copula}
It should be noted that in the proposed LHSD method, the joint distribution of the inputs, $F$, 
and the corresponding marginal distributions $F_k(x_k|x_1,x_2, \cdots ,x_{k-1})$,  $k=2,3, \cdots ,K$ are assumed to be known. However, in some applications the joint distributions and the corresponding conditional distributions may be unknown. In those situations, we propose to use a copula-based method to construct the joint distribution from the known marginal distributions. 
Let the marginal cumulative distribution function of $X_k$ be $F_k(x_k), \ k=1,2, \cdots ,K$. Then according to Sklar's theorem \citep{sklar}, there exists a copula $C : [0,1]^K \rightarrow [0,1]$, such that the joint  density function  of $X_1, X_2, \cdots ,X_K$ can be written as
\beq
F(x_1,x_2, \cdots ,x_k)=C(F_1(x_1), F_2(x_2), \cdots ,F_K(x_K)).
\eeq
Let us define $U_k=F_k(X_k), \ k=1,2, \cdots ,K$. 
If $F$ is continuous then  $C$ is uniquely written as
\beq
C(u_1,u_2, \cdots ,u_K)=F(F_1^{-1}(u_1), F_2^{-1}(u_2), \cdots ,F_K^{-1}(u_K)).
\eeq
For $k=2,3, \cdots ,K$, the $k$-th conditional copula function is defined as 
\begin{align}
\label{condc}
C_k(u_k|u_1, \cdots ,u_{k-1}) 
& = P(U_k\leq u_k|U_1=u_1, \cdots ,U_{k-1}=u_{k-1}) \nonumber\\ & =\frac{\partial^{k-1}C_k(u_1, u_2, \cdots ,u_k)}{\partial u_1, \cdots ,\partial u_{k-1}}, 
\end{align}
where $C_k(u_1,u_2, \cdots ,u_k)=C(u_1, u_2, \cdots ,u_k,1,1, \cdots ,1)$, is the $k$-dimensional marginal copula function. 
 The conditional distribution of  $X_k$ given $X_1=x_1, X_2=x_2, \cdots, X_{k-1}=x_{k-1}$ is defined as
\begin{align}
    \label{condif}
& F(x_k|X_1=x_1, X_2=x_2, \cdots X_{k-1}=x_{k-1}) \nonumber\\
& \ \ \ =\frac{\partial^{k-1} }{\partial u_1, \cdots ,\partial u_{k-1}}C_k(u_1, \cdots ,u_k) \Bigl|_{u_1=F_1(x_1), \cdots ,u_k=F_k(x_k)}.
\end{align}
By choosing a suitable copula and then finding the conditional distributions from equation \eqref{condif}, one can use the same LHSD algorithm (Algorithm \ref{lhsd}) to sample from the joint distribution.
Alternatively, one can also use the conditional copulas in equation \eqref{condc} to sample from the joint distribution using  Algorithm \ref{lhsdc}.

\begin{algorithm}
\caption{LHSD using Copula}
\label{lhsdc}
\begin{enumerate}
\item Suppose the components of ${\mathbf{Z}}$ are i.i.d. $U(0,1)$. Using  Algorithm \ref{lhs} generate $z_{jk},\ j=1,2,\cdots N, \ k=1,2.\cdots, K$, a LHS of size $N$ from ${\mathbf{Z}}$.
\item Sequentially get ${\mathbf{u}}_j =(u_{j1}, u_{j2}, \cdots, u_{jK})^T$, the $j$-th random sample from the joint copula for $j=1,2, \cdots,N$, using the inverse conditional copula functions,
\begin{itemize}
\item $u_{j1} = z_{j1}$,
\item $u_{j2} = C_2^{-1}(z_{j2}|u_{j1})$,
\item $u_{j3} = C^{-1}_3(z_{j3}|u_{j1}, u_{j2}),$ \\
 \vdots
\item $u_{jK} = C^{-1}_K(z_{jK}|u_{j1}, \cdots ,u_{jK-1})$,
\end{itemize}

\item Obtain ${\mathbf{x}}_j =(x_{j1}, x_{j2}, \cdots, x_{jK})^T$, the $j$-th random sample from $\mathbf{X}$ for $j=1,2, \cdots,N$, using the inverse distribution function,
$x_{jk} = F^{-1}_K(u_{jk})$, for $k=1, 2, \cdots, K$.
\end{enumerate}
\end{algorithm}

There are several families of copula functions; and the most
commonly used families are elliptical copula family -- such as
Gaussian copula and t-copula, and Archimedean copula
family -- such as Gumbel copula, Clayton copula, Frank
copula etc. Different copula functions are suitable
to measure different dependency structures among the random
variables. For example, Gumbel copula is suitable to
demonstrate the dependency structure with upper tail
dependence, and Frank copula is adopted to measure the
symmetric dependency structure.

Given a data set, choosing a suitable copula function that fits
the data is an important but difficult problem. Since the
real data generation mechanism is unknown, it is possible that
several candidate copula functions may fit the data reasonably
well; on the other hand,  none of the candidate may give a reasonable fit. 
However, the common procedure to select a suitable copula function
contains the following three steps: (i) choose several commonly used
copula functions; (ii) for a given data set,
estimate  parameters of the selected copula functions by
some methods, such as the maximum likelihood estimation
method; and finally, (iii) select the optimal copula function. A divergence based shortest distance measure between the empirical distribution function and the estimated copula function is used to select the optimal copula function. 

It is to be noted that Algorithm \ref{lhsdc} can also be used to sample from any joint distribution when the conditional distributions are not known or difficult to calculate in the closed form, but the conditional copulas are easier to find.

\section{Large Sample Properties of the LHSD} \label{sec:large_sample}
Suppose we are interested in a measurable function $h: \mathbb{R}^K \rightarrow \mathbb{R}$. Our goal is to estimate $\tau = E(h(\mathbf{X}))$, where $\mathbf{X}$ is a random vector with distribution function $F(\cdot)$.  We assume that $h(\cdot)$ has a finite second moment. Based on $\mathbf{X}_1, \mathbf{X}_2, \cdots, \mathbf{X}_N$, a LHSD  of size $N$, the estimator of $\tau$ is given by $\hat{\tau}=N^{-1}\sum_{j=1}^N h(\mathbf{X}_j)$. In this section, we discuss about the asymptotic properties of $\hat{\tau}$. 

We define  $\mathbf{Z} = (Z_1, Z_2, \cdots, Z_K)^T$ by the following transformations:
\begin{itemize}
\item $Z_{j1} = F_1(X_{j1})$,
\item $Z_{j2} = F_2(X_{j2}|X_1=x_1)$,
\item $Z_{j3} = F_3(X_{j3}|X_1=x_1, X_2=x_2)$,\\
\vdots 
\item $Z_{jK} = F_K(X_{jK}|X_1=x_1, X_2=x_2, \cdots, X_{k-1}=x_{k-1})$,
\end{itemize}
where  $F_k(\cdot|X_1=x_1, X_2=x_2, \cdots, X_{k-1}=x_{k-1})$ is the conditional distribution function of $X_k$ given $X_1=x_1, X_2=x_2, \cdots, X_{k-1}=x_{k-1}$ for $k= 2, 3,  \cdots, K$. Suppose the transformation can be written as $\mathbf{Z} = \mathbf{m}(\mathbf{X})$, where  ${\mathbf{m}}: \mathbb{R}^K \rightarrow \mathbb{R}^K$. Then, we have
\begin{equation}
h({\mathbf{X}}) = h({\mathbf{m}}^{-1}({\mathbf{Z}})) \equiv h^*({\mathbf{Z}}),
\label{inv_transformation}
\end{equation}
where $h^*: \mathbb{R}^K \rightarrow \mathbb{R}$. Thus,  $\tau = E(h({\mathbf{X}})) = E(h^*(\mathbf{Z}))$.

Note that the components of $\mathbf{Z}$ are i.i.d. $U(0,1)$ variables. Once we transform $\mathbf{X}$ to $\mathbf{Z}$ and subsequently $h(\cdot)$ to $h^*(\cdot)$, our problem simplifies to the estimation of $h^*(\mathbf{Z})$, where the components of $\mathbf{Z}$ are i.i.d. $U(0,1)$ variables. Thus, the theoretical properties of $\hat{\tau}$ can be derived obtained from the work of \cite{stein1} and \cite{Owen}.

Let us assume that $\int_{A^K} h^{*2}({\mathbf{z}}) d{\mathbf{z}} < \infty$, where $A^K = \{{\mathbf{z}} \in \mathbb{R}^K : ||{\mathbf{z}}|| \leq 1 \}$. Following \cite{stein1} we  decompose $h^*$ as:
\begin{equation}
h^*({\mathbf{z}}) = \tau + \sum_{k=1}^K \alpha_k(z_k) + r({\mathbf{z}}),
\label{decomp}
\end{equation}
where 
\begin{equation}
\alpha_k(z_k) = \int_{A^{K-1}} (h^*({\mathbf{z}}) - \tau) d{\mathbf{z}}_{-k},
\label{def:alpha}
\end{equation}
${\mathbf{z}}_{-k}$ being the vector $\mathbf{z}$ without the $k$-th component. Here,  $\alpha_k(z_k)$ is called as the ``main effect'' function of the $k$-th component of $\mathbf{Z}$, and $r(\mathbf{z})$, defined by subtraction in equation (\ref{decomp}), is the ``residual from additivity'' of $h^*$. It is clear from the definition that
\begin{equation}
\int_0^1 \alpha_k(z_k) dz_k = 0 \mbox{ and } \ \
\int_{A^{K-1}} r({\mathbf{z}}) d{\mathbf{z}}_{-k} = 0,
\end{equation}
for all $k=1,2, \cdots, K$.
The variance of $\hat{\tau}$ under the LHSD is derived from the following theorem.
\begin{theorem}
As $N \rightarrow \infty$
\begin{equation}
{\rm{Var_{LHSD}}}(\hat{\tau}) = \frac{1}{N} \int_{A^K} r^2({\mathbf{z}}) d{\mathbf{z}} + o\left(\frac{1}{N}\right).
\label{var_LHS}
\end{equation}
\label{theorem:var_lhsd}
\end{theorem}
The proof of the theorem is given in Appendix, and it directly follows from \cite{stein1}. 
Note that the variance of  $\hat{\tau}$ under the SRS is given by
\begin{equation}
{\rm{Var_{SRS}}}(\hat{\tau}) = \frac{1}{N} \int_{A^K} r^2({\mathbf{z}}) d{\mathbf{z}} + \frac{1}{N} \sum_{k=1}^K \int_{A^1} \alpha_k^2(z_k) dz_k + o\left(\frac{1}{N}\right).
\label{var_SRS}
\end{equation}
It shows that the variance of $\hat{\tau}$ in the LHSD is smaller than that of the SRS unless all the main effects $\alpha_k$ are equal to zero. For a fixed $N$, the variance of $\hat{\tau}$ is reduced by $\frac{1}{N} \sum_{k=1}^K \int_{A^1} \alpha_k^2(z_k) dz_k $ if the LHSD is used instead of the SRS.  Therefore, if the main effects are large, the reduction in variance can be substantial. One may notice that the function $\alpha_k$ in equation (\ref{decomp}) depends on the order of $\mathbf{X}$ in the transformation of $\mathbf{Z}$; and hence it is possible to minimize the variance by taking an appropriate order of $\mathbf{X}$. Our conjecture is that if the components of $\mathbf{X}$ are arranged in the decreasing order of their sensitivity in estimating function $h$, then there will be a maximum reduction of variance. This will be addressed in more detail in our future research. 
Now, the asymptotic distribution of $\hat{\tau}$ under the LHSD is obtained from the following theorem.
\begin{theorem}
The asymptotic distribution of $\hat{\tau}$ under the LHSD is given by
\begin{equation}
N^{1/2} (\hat{\tau} - \tau ) \ {\overset{a}{\sim}} \ N\left(0, \int_{A^K} r^2({\mathbf{z}}) d{\mathbf{z}} \right).
\end{equation} \label{theorem:asymp_dist}
\end{theorem}
After taking transformation from $\mathbf{X}$ to $\mathbf{Z}$, the proof of the theorem directly follows from \cite{Owen}. An outline of the proof is given in Appendix. Using the above theorem, we can construct a confidence interval for $\tau$ provided we have an estimate of $\int_{A^K} r^2({\mathbf{z}}) d{\mathbf{z}} $. For this reason, we must estimate the additive main effects $\alpha_k$ from the data. One may use a non-parametric regression for this purpose. It assumes some degree of smoothness for $\alpha_k$ in case of application. Equation (\ref{decomp}) shows that a non-parametric regression of  $\mathbf{Y} = h^*(\mathbf{Z})$ on $Z_k$ provides an estimate $\hat{g}_k$ of the function $g_k = \tau + \alpha_k$. So, $\alpha_k$ is estimated by $\hat{\alpha}_k = \hat{g}_k - \bar{Y}$, where
$$ \bar{Y} = \frac{1}{N} \sum_{j=1}^N \mathbf{Y}_j =\frac{1}{N} \sum_{j=1}^N h^*(\mathbf{Z}_j)   = \frac{1}{N} \sum_{j=1}^N h(\mathbf{X}_j) = \hat{\tau}. $$
Then, the residual for the $i$-th observation is estimated as
$$ \hat{r}_i = Y_i - \bar{Y} - \sum_{k=1}^K \hat{\alpha}_k(Z_{ik}) .$$
Finally, the variance $\int_{A^K} r^2({\mathbf{z}}) d{\mathbf{z}}$ is estimated by $\hat{r}_i^2/N$. The theoretical properties of the variance estimate are discussed in \cite{Owen}.

\section{Simulation Studies} \label{sec:simulation}
To compare the performance of our proposed LHSD and LHSD$_c$ methods with the other existing methods, we perform two simulation studies. In the first example, the set of input variables is assumed to follow a multivariate normal distribution; the second example takes  Gumbel's bivariate logistic distribution. The estimated quantity is the mean of a non-linear function of the output variables. 

\subsection{LHS from a Multivariate Normal Distribution}
 We assume that input variable $\mathbf{X} = (X_1, X_2, \cdots, X_K)^T$ follows a multivariate normal distribution, $N_K({\mathbf{\mu, \Sigma}})$, where ${\mathbf{\mu}}=(\mu_1, \mu_2, \cdots, \mu_K)^T$ is the mean vector, and ${\mathbf{\Sigma}} = ((\sigma_{ij}))_{K\times K}$ is the covariance matrix. Now, $X_1 \sim N(\mu_1, \sigma_{11})$, and for $k=2,3, \cdots, K$ the conditional distribution of the $k$-th component of $\mathbf{X}$ is given by $X_k|x_1, \cdots ,x_{k-1} \sim N(\mu_k^*, \sigma_k^{*2})$, where %
\begin{eqnarray}
    \mu_k^* &=& \mu_k + \mathbf{\sigma}_k^T \mathbf{\Sigma}_{kk}^{-1}  \left( \mathbf{x}^{(k)} - \mathbf{\mu}^{(k)} \right),  \\
    \sigma_k^{*2} &=& \sigma_{kk} - \mathbf{\sigma}_k^T \mathbf{\Sigma}_{kk}^{-1} \mathbf{\sigma}_k     .
\end{eqnarray}
Here, $\mathbf{x}^{(k)}=(x_1, x_2, \cdots, x_{k-1})^T$, $\mathbf{\mu}^{(k)}=(\mu_1, \mu_2, \cdots, \mu_{k-1})^T$,  $\mathbf{\sigma}_k = (\sigma_{k1}, \sigma_{k2}, \cdots, \sigma_{k(k-1)})^T$, and $\mathbf{\Sigma}_{kk}$ is the $(k-1)\times (k-1)$ matrix containing the first $(k-1)$ rows and $(k-1)$ columns of $\mathbf{\Sigma}$. 
Then, using the proposed LHSD algorithm (Algorithm \ref{lhsd}), we get the samples from the joint distribution using the following steps:

\begin{enumerate}

\item Using  Algorithm \ref{lhs}, we generate $z_{jk},\ j=1,2,\cdots N, \ k=1,2.\cdots, K$, a LHS of size $N$ from ${\mathbf{Z}}$ whose components  are from i.i.d. $U(0,1)$.

\item For $j=1,2, \cdots, N$ the LHSD from $\mathbf{X}$ is given by
\begin{itemize}
\item $x_{j1} = \Phi^{-1}(z_{j1}; \mu_1, \sigma_{11})$,
\item $x_{j2} = \Phi^{-1}(z_{j2}; \mu_2^*, \sigma_2^{*2})$,
\item $x_{j3} = \Phi^{-1}(z_{j3}; \mu_3^*, \sigma_3^{*2})$, \\
\vdots
\item $x_{jK} = \Phi^{-1}(z_{jK}; \mu_K^*, \sigma_K^{*2})$,
\end{itemize}
where  $\Phi^{-1}(\cdot; \mu^*, \sigma^{*2})$ is the inverse distribution function of the normal distribution with mean $\mu^*$, and standard deviation $\sigma^*$.
\end{enumerate}

\begin{figure}
\centering
\includegraphics[height=8cm, width=12cm]{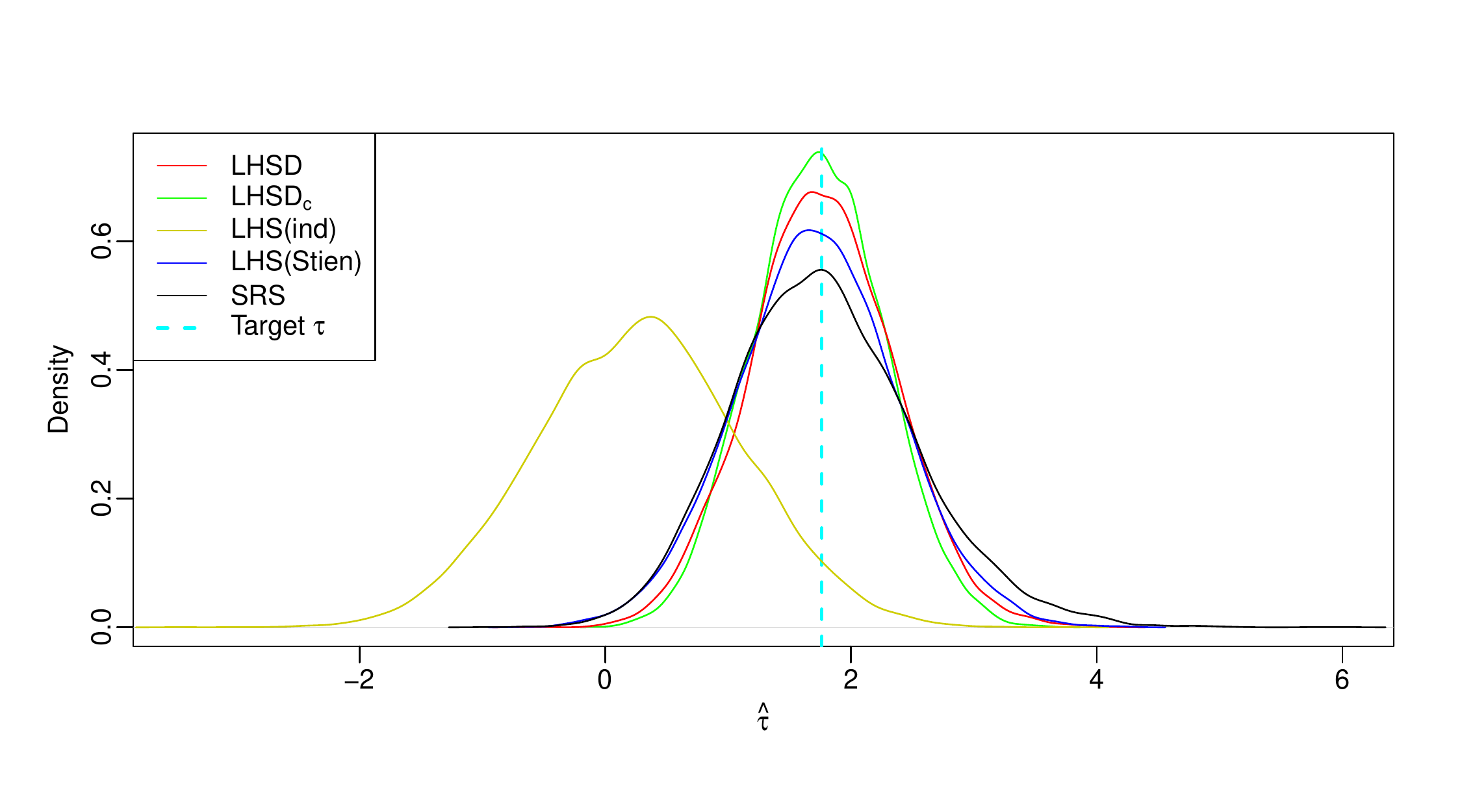}
\caption{The empirical density plot of $\hat{\tau}$ using five different sampling methods, where samples are drawn from a multivariate normal distribution.}
\label{fig:LHS_Normal}
\end{figure}


In this simulation study, the input variable $\mathbf{X}$ is generated from a four-dimensional multivariate normal distribution $N_4(\mu, \Sigma)$. The parameters of the distribution, i.e. $\mu$ and $\Sigma$, are randomly chosen, and then they are kept fixed for all replications. Here, $\mu$ is generated from the standard normal distribution. 
For the covariance matrix $\Sigma$,  first we generated a $4\times 4$ matrix $P$ from the standard normal distribution, then we took $\Sigma = P^TP$. We have generated random samples using five different methods -- (i) LHSD: LHS using our proposed method given in Algorithm \ref{lhsd}, (ii) $\rm{LHSD_c}$: centered LHSD as proposed in Section \ref{sec:lhsd_cen}, (iii)  LHS(ind): LHS assuming that the components of $\mathbf{X}$ are independent as given in Section \ref{LHS_ind}, (iv) LHS(Stien): rank-based LHS  proposed by \cite{stein1} and (v) SRS: the simple random sample.  We used a fixed sample of size $N=30$ in all five cases and replicated the simulations $10,000$ times. We are interested in estimating $\tau = E(h(\mathbf{X}))$, where
\begin{equation}
h(\mathbf{X}) = x_1 + \frac{x_2 x_3}{2} - x_2 \log(|x_1|) + \exp\left(\frac{x_4}{4}\right),
\end{equation}
and $\mathbf{X} = (x_1, x_2, x_3, x_4)^T$. We estimate $\tau$ by $\hat{\tau}=N^{-1}\sum_{j=1}^N h(\mathbf{X}_j)$ using these five sampling procedures. The empirical densities of $\hat{\tau}$ are plotted in Figure \ref{fig:LHS_Normal}. 
We also repeated the experiment by taking sample sizes as $20$, $75$, and $100$.
The bias, variance and MSE of $\sqrt{N}\hat{\tau}$ for all the five sampling methods are presented in Table \ref{tab:LHS_Normal}. The true value of $\tau = E(h(\mathbf{X}))$ is calculated based on a SRS of size $N=10,000$. The results show that LHS(ind) has a large bias, as it ignores the dependency structure among the components of $\mathbf{X}$. All other methods give very small bias. The variance of the SRS is very large, and for this reason, the MSE is also very large. On the other hand, our LHSD-based methods give the lowest MSE among all the methods. As mentioned in the previous section, the centered LHSD introduces a small bias, but it substantially reduces the variance of the estimator. As a result, it outperforms the other sampling schemes.

\begin{table}[]
    \centering
    \begin{tabular}{l|rrr|rrr}
\multirow{2}{*}{Methods}   &  Bias   &   Variance     &      MSE &  Bias   &   Variance     &      MSE \\
\cline{2-7}  &     \multicolumn{3}{c|}{$N=20$}   &   \multicolumn{3}{c}{$N=30$}  \\
         \hline
       LHSD &   0.08 & 10.55 & 10.56 & 0.03 & 10.22 & 10.22\\
       $\rm{LHSD_c}$ & -0.20 & 7.19 & 7.23 & -0.15 & 8.24 & 8.26\\
   LHS(ind) &  -6.59 & 21.24 & 64.61 & -8.15 & 21.03 & 87.45\\
 LHS(Stien) & -0.39 & 13.06 & 13.21 & -0.22 & 12.81 & 12.86\\
        SRS & -0.03 & 16.50 & 16.50 & 0.12 & 16.75 & 16.77\\
\hline
Methods &     \multicolumn{3}{c|}{$N=75$}   &   \multicolumn{3}{c}{$N=100$}  \\
\hline
LHSD & 0.04 & 9.18 & 9.18 & -0.42 & 9.09 & 9.27\\
 $\rm{LHSD_c}$ &  -0.09 & 8.25 & 8.26 & -0.56 & 8.10 & 8.42\\ 
   LHS(ind) & -12.77 & 20.03 & 183.16 & -15.29 & 19.40 & 253.24\\ 
   LHS(Stien) & -0.24 & 11.96 & 12.02 & -0.60 & 12.13 & 12.49\\ 
   SRS &-0.06 & 16.55 & 16.55 & -0.41 & 16.64 & 16.81\\ 
   \hline
     \end{tabular}
    \caption{The bias, variance and MSE of $\sqrt{N}\hat{\tau}$ using five different sampling methods, where samples are drawn from a multivariate normal distribution.}
    \label{tab:LHS_Normal}
\end{table}

Table \ref{tab:LHS_Normal} shows that the bias and variance of $\sqrt{N}\hat{\tau}$ for the LHSD  are stable as the sample size increases,  so it demonstrates that the estimator is $\sqrt{N}$-consistent as proved in Theorem \ref{theorem:asymp_dist}.
We also calculated the theoretical variances of the LHSD and SRS methods using equations (\ref{var_LHS}) and (\ref{var_SRS}), respectively, for sample size $30$. We generated a large SRS of size $N=10,000$  and estimated $\alpha_k(z_k)$ and $r(\mathbf{z})$ as discussed at the end of Section \ref{sec:large_sample}. The k-nearest neighbor algorithm is used in the non-parametric regression. The theoretical variances of $\sqrt{N}\hat{\tau}$ for the LHSD and SRS are obtained as $9.45$ and $17.60$, respectively. These values are also very close to the simulation results as given in Table \ref{tab:LHS_Normal}.

\begin{table}[]
    \centering
    \begin{tabular}{l|rr|rr}
\multirow{2}{*}{Methods}   &  Mean   &   Variance      &  Mean   &   Variance     \\
\cline{2-5}  &     \multicolumn{2}{c|}{$N=20$}   &   \multicolumn{2}{c}{$N=30$}  \\
         \hline
       LHSD & 13.372 & 0.067 & 12.962 & 0.034\\
       $\rm{LHSD_c}$ & 13.190 & 0.022 & 12.843 & 0.014\\
   LHS(ind) & 59.854 & 138.219 & 59.742 & 89.030\\
 LHS(Stien) & 17.115 & 3.746  & 15.592 & 1.546\\
        SRS &13.114 & 0.511 & 12.760 & 0.328\\
\hline
Methods &     \multicolumn{2}{c|}{$N=75$}   &   \multicolumn{2}{c}{$N=100$}  \\
\hline
LHSD & 12.287 & 0.009 &  12.123 & 0.006\\
 $\rm{LHSD_c}$ & 12.237 & 0.006 & 12.085 & 0.005\\ 
   LHS(ind) & 59.299 & 33.388 & 59.118 & 23.950 \\ 
   LHS(Stien) & 13.464 & 0.247 & 13.037 & 0.148 \\ 
   SRS & 12.156 & 0.126 & 12.016 & 0.093\\ 
   \hline
     \end{tabular}
    \caption{The mean and variance of the KL-divergence measure for five different sampling schemes, where samples are drawn from a multivariate normal distribution.}
    \label{tab:LHS_Normal_KL}
\end{table}

Now, we present a comparison study to see how well the different sampling schemes represent the target distribution. We need a goodness-of-fit measure to calculate the deviation of a random sample from the target distribution $N_4(\mu, \Sigma)$. The Kullback-Leibler (KL) divergence is used for this purpose (\citealp{kl1951}). Suppose $f(\cdot)$ is the probability density function (p.d.f.) of the target distribution, and $f_n(\cdot)$ is the empirical p.d.f. Then, the KL-divergence between two densities is defined as
\begin{equation}
    {\rm{KL}}(f_n, f) = \int_x f_n(x) \log \frac{f_n(x)}{f(x)} dx 
    = \int_x f_n(x) \log f_n(x) dx -  \int_x f_n(x) \log  f(x) dx.  
    \label{kl}
\end{equation}
The first term in equation (\ref{kl}) is called the Shannon entropy function, which is estimated using `entropy' function of `KNN' package in R. The second term in equation (\ref{kl}) is simply estimated by $\sum_{i=1}^N \log f(x_i)$. A small value of the KL-divergence indicates that the densities $f_n$ and $f$ are close to each other. If the two densities are identical, then the KL-divergence is zero, otherwise, it is always positive. 
The mean and variance of the  KL-divergence over 10,000 replications are presented in Table \ref{tab:LHS_Normal_KL} for different sample sizes. The SRS gives the minimum mean, as it  generates  unbiased random samples from the target distribution. Our LHSD-based sampling schemes are also giving the mean KL-divergence very close to the SRS, and their variances are  around ten times smaller than the SRS. The mean from the rank based LHS method is very high for the small sample sizes, which indicates that there is a large bias in sampling from the target distribution. We also noticed that the variance of the KL-divergence using rank based LHS is considerably high in this example.

\subsection{LHS from a Bivariate Logistic Distribution}
In this section, we assume that input variable $\mathbf{X}$ follows the Gumbel's bivariate logistic distribution, where the joint distribution function is given by
$$ F(x_1, x_2) = \frac{1}{1+ \exp(-x_1)+ \exp(- x_2)}, \ \ x_1, \ x_2 \geq 0.$$
The marginal distribution of $X_k$ is given by $F_k(x_k) =  \frac{1}{1+ \exp(- x_k)}, \ k=1,2$, The corresponding bivariate copula function is given by $C(u_1,u_2)=\frac{u_1u_2}{u_1+u_2-u_1u_2}$, $u_1,u_2 \in [0,1]$. Since the copula is well known, we use Algorithm \ref{lhsdc} to sample from the joint distribution using LHSD.
The conditional copula is given by $C_2(u_2|u_1)=\frac{u_2^2}{(u_1+u_2-u_1u_2)^2}$. The samples from the joint distribution are obtained using the following steps. 

\begin{enumerate}

\item Using  Algorithm \ref{lhs}, we generate $z_{jk},\ j=1,2,\cdots N, \ k=1,2.\cdots, K$, a LHS of size $N$ from ${\mathbf{Z}}$ whose components  are from i.i.d. $U(0,1)$.

\item For $j = 1,2, \cdots ,N$, the corresponding samples from the joint copula is given by
\begin{itemize}
    \item $u_{j1}=z_{i1}$
    \item $u_{j2} = C_2^{-1}(z_{j2}|u_{j1}) = \frac{u_{j1}\sqrt{z_{j2}}}{1-\sqrt{z_{j2}}+u_{j1}\sqrt{z_{j2}}}$.
\end{itemize}
 
\item For $j=1,2,\cdots N$ the LHSD from $\mathbf{X}$ is given by


$x_{jk} = F_k^{-1}(u_{jk}) = -\log\left(\frac{1-u_{jk}}{u_{jk}}\right),\  k=1, 2.$



\end{enumerate}

We are interested in estimating  $\tau = E(h(\mathbf{X}))$, where
$
h(\mathbf{X}) = x_1 - x_2 + x_2 \log(|x_1|).
$
We generate random samples using the same five methods discussed in the previous example. Here, we use  $N=30$ in all five cases and replicate the simulations $10,000$ times. The empirical density functions of $\hat{\tau}$ using these  methods are presented in Figure \ref{fig:LHS_Exponential_plot}. The biases and MSEs are given in Table \ref{fig:LHS_Exponential_plot} for four different  sample sizes $N=20$, 30, 75 and 100. It is seen that the centered LHSD performs best. The rank based LHS method and LHSD have similar performance, but LHSD is slightly better in terms of the MSE. The LHS assuming independence produced a very large bias, and the SRS produced a large variance as expected.

\begin{figure}
\centering
\includegraphics[height=8cm, width=12cm]{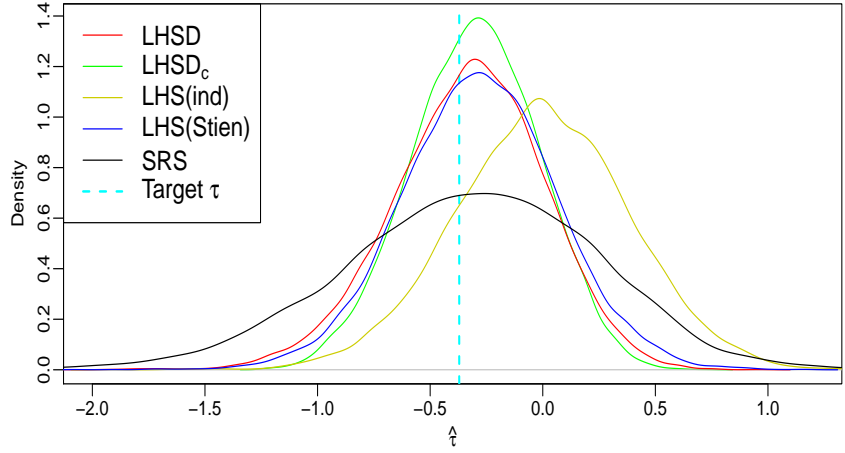}
\caption{The empirical density plot of $\hat{\tau}$ using five different sampling methods, where samples are drawn from the  bivariate logistic distribution.}
\label{fig:LHS_Exponential_plot}
\end{figure}


\begin{table}[]
    \centering
    \begin{tabular}{l|rrr|rrr}
\multirow{2}{*}{Methods}   &  Bias   &   Variance     &      MSE &  Bias   &   Variance     &      MSE \\
\cline{2-7}  &     \multicolumn{3}{c|}{$N=20$}   &   \multicolumn{3}{c}{$N=30$}  \\
         \hline
       LHSD  & 0.24 & 3.74 & 3.79 &  0.22 & 3.42 & 3.46\\
       $\rm{LHSD_c}$ & 0.43 & 2.29 & 2.47 & 0.39 & 2.38 & 2.53\\
   LHS(ind)  & 1.70 & 4.89 & 7.77 & 2.01 & 4.52 & 8.55\\
 LHS(Stien) & 0.50 & 3.70 & 3.95  & 0.48 & 3.61 & 3.84\\
        SRS & 0.26 & 9.74 & 9.81  & 0.26 & 9.67 & 9.74\\
\hline
Methods &     \multicolumn{3}{c|}{$N=75$}   &   \multicolumn{3}{c}{$N=100$}  \\
\hline
LHSD  &-0.23 & 3.19 & 3.24 & 0.06 & 3.14 & 3.15\\
 $\rm{LHSD_c}$ & -0.09 & 2.55 & 2.55 & 0.16 & 2.76 & 2.79\\ 
   LHS(ind) & 2.61 & 4.45 & 11.24 & 3.33 & 4.38 & 15.45\\ 
   LHS(Stien) &-0.11 & 3.26 & 3.27 & 0.20 & 3.21 & 3.25\\ 
   SRS & -0.23 & 9.92 & 9.97 & 0.06 & 9.90 & 9.91\\ 
   \hline
     \end{tabular}
    \caption{The bias, variance and MSE of $\sqrt{N}\hat{\tau}$ using five different sampling methods, where samples are drawn from the bivariate logistic distribution.}
    \label{tab:LHS_Exponential}
\end{table}

\begin{table}[]
    \centering
    \begin{tabular}{l|rr|rr}
\multirow{2}{*}{Methods}   &  Mean   &   Variance      &  Mean   &   Variance     \\
\cline{2-5}  &     \multicolumn{2}{c|}{$N=20$}   &   \multicolumn{2}{c}{$N=30$}  \\
         \hline
       LHSD &7.732 & 0.033 &  7.682 & 0.016 \\
       $\rm{LHSD_c}$ & 7.635 & 0.007 & 7.617 & 0.004 \\
   LHS(ind) &8.312 & 0.074 & 8.277 & 0.042 \\
 LHS(Stien) &7.828 & 0.093 & 7.757 & 0.058 \\
        SRS  &7.564 & 0.294  & 7.531 & 0.195 \\
\hline
Methods &     \multicolumn{2}{c|}{$N=75$}   &   \multicolumn{2}{c}{$N=100$}  \\
\hline
LHSD & 7.639 & 0.004  & 7.636 & 0.003\\
 $\rm{LHSD_c}$ & 7.615 & 0.002 & 7.617 & 0.002\\ 
   LHS(ind) &8.257 & 0.012 & 8.258 & 0.008 \\ 
   LHS(Stien) &7.681 & 0.022 & 7.673 & 0.017 \\ 
   SRS & 7.536 & 0.078  & 7.536 & 0.057 \\ 
   \hline
     \end{tabular}
    \caption{The mean and variance of the Kullback-Leibler divergence measure for five different sampling methods, where samples are drawn from the bivariate logistic distribution. }
    \label{tab:LHS_Exponential_KL}
\end{table}

Like the previous example, we also estimated the KL-divergence using equation (\ref{kl}). The result is shown in Table \ref{tab:LHS_Exponential_KL}. On an average, the SRS produces the minimum KL-divergence as the sampling scheme is unbiased. Both LHSD and LHSD$_c$ gives slightly higher mean values than SRS, however they are lower than the independent LHS and the ranked based LHS methods. The variance of the KL-divergence using LHSD$_c$ method is very low compared to the other methods.  

\section{Real Data Example: River Flood Inundation} \label{sec:river}
We  now illustrate our sampling methodology for a real application related to the river flood inundation model (\citealp{MR2807168}). The study case concerns an industrial site located near a river and protected from it by a dyke. When the river height exceeds  one of the dyke, flooding occurs. The goal is to study the water level with respect to the dyke height. 
The model considered here is based on a crude simplification of the 1D hydro-dynamical equations of Saint Venant under  assumptions of uniform and constant flow rate and large rectangular sections. It consists of an equation that involves the characteristics of the river stretch given by
\begin{eqnarray}
\label{rivereq}
S &=& Z_v + H - H_d - C_b, \\
H &=& \left( \frac{Q}{BK_s \sqrt{\frac{Z_m - Z_v}{L}}} \right)^{0.6}.
\end{eqnarray}
The model output $S$ is the maximal
annual overflow (in meters) and $H$ is the maximal annual height of the river (in meters). The inputs of the model (8 inputs) together  with their probability distributions are defined in Table \ref{tab:river}. Among the input variables of the model, dyke height $(H_d)$ is a design parameter; its variation range corresponds to a design
domain. The randomness of the other variables is due to their spatio-temporal variability, our ignorance of their true value or some inaccuracies of their estimation. The river flow model is illustrated in Figure \ref{fig:River}.

\begin{table}[]
    \centering
    \setlength\tabcolsep{1.5pt} 
    \begin{tabular}{lccc}
         Input & Description & Unit & Probability distribution\\
         \hline
$Q$ & Maximal annual flowrate &  $m^3/s$ &  Truncated Gumbel $G(1013, 558)$ on $[500, 3000]$\\
$K_s$ &  Strickler coefficient &  -- &  Truncated normal $N(30, 8)$ on $[15,\infty)$\\
$Z_v$ &  River downstream level &  m &  Triangular $T(49, 50, 51)$\\
$Z_m$ &  River upstream level &  m  & Triangular $T(54, 55, 56)$\\
$H_d$ &  Dyke height &  m &  Uniform $U(7, 9)$\\
$C_b$ &  Bank level &  m &  Triangular $T(55, 55.5, 56)$\\
L &  Length of the river stretch &  m &  Triangular $T(4990, 5000, 5010)$\\
B &  River width &  m &  Triangular $T(295, 300, 305)$\\
\hline
     \end{tabular}
    \caption{Input variables of the flood model and their probability distributions.}
    \label{tab:river}
\end{table}

\begin{figure}
\centering
\includegraphics[height=8cm, width=12cm]{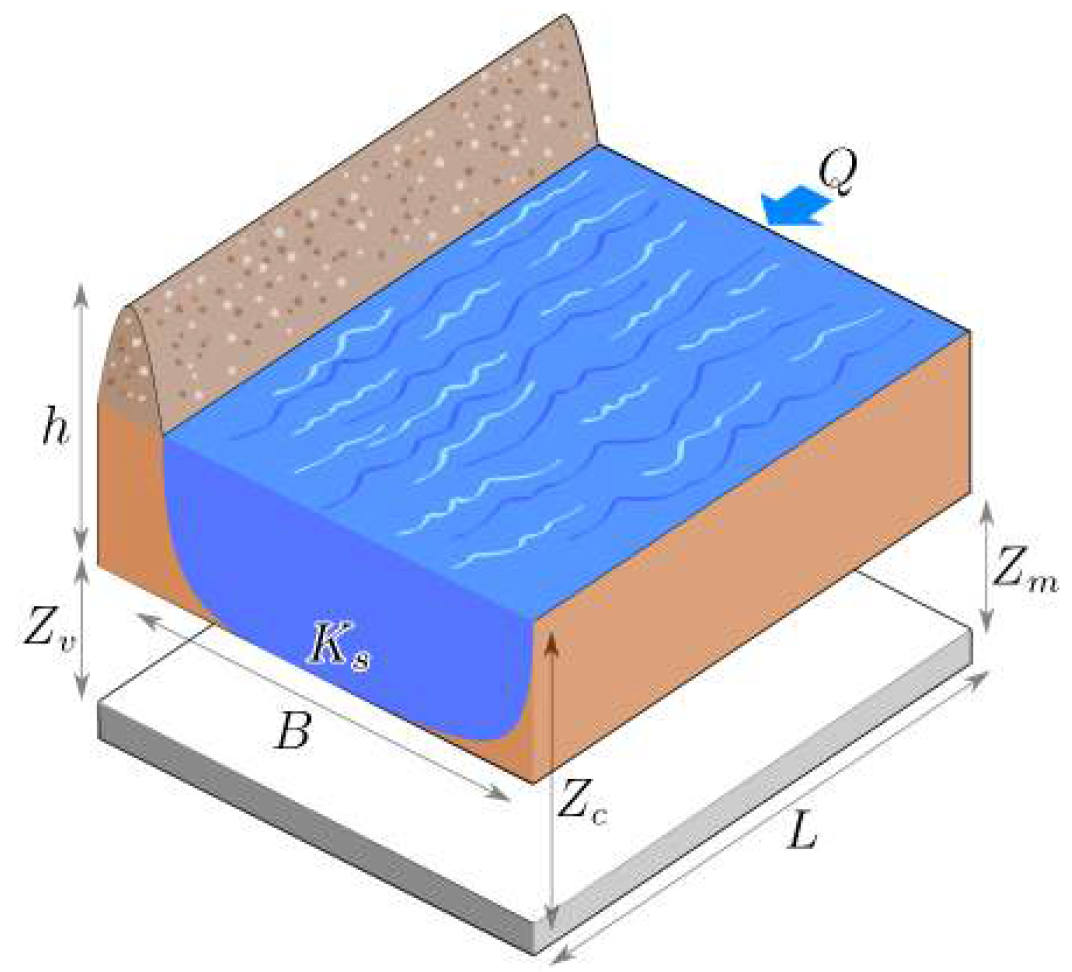}
\caption{Different inputs for the river flood inundation model.}
\label{fig:River}
\end{figure}

From the previous studies it is concluded that the inputs maximal annual flow-rate ($Q$) and Strickler coefficient ($K_s$) are correlated with correlation coefficient $\rho_{Q,K} = 0.5$ (see \citealp{MR3020270}). This correlation is admitted in real case as it is known that the friction coefficient increases with the flow rate. Similarly, the river downstream $(Z_v)$ and river upstream level $(Z_m)$ are known to be correlated with correlation coefficient $\rho_{Z_m,Z_v} = 0.3$. The correlation coefficient between the length of the river $(L)$ and breadth of the river $(B)$ are also known to be $\rho_{L,B} = 0.3$. The other input variables are assumed to be uncorrelated.

The quantity of interest is the mean maximal annual overflow $\tau=E(S)$. We compare our proposed LHSD-based methods with the SRS, LHS(ind), LHS(Stien) in estimating $\tau$. Although the marginal distribution of the inputs $X=(Q, K_s, Z_v, Z_m, H_d, C_b, L, B)^T$ are known, the joint distribution is not known. We construct a joint distribution using the multivariate normal copula with the known pairwise correlation coefficients as described in Section \ref{sec:copula}. The choice of such copula is appropriate here as it would retain the known correlation structure in the inputs. The multivariate normal copula function for $K$ variables is given by
\beq
C(u_1,u_2, \cdots ,u_K; \Sigma)=\Phi_K(\Phi^{-1}(u_1), \Phi^{-1}(u_2), \cdots, \Phi^{-1}(u_K);\Sigma),
\eeq
where  $\Phi_K(\cdot;\Sigma)$ is the joint c.d.f. of the multivariate normal distribution with covariance matrix $\Sigma$, and $\Phi(\cdot)$ is the c.d.f. of a standard normal distribution. In this example, $K=8$, and $\Sigma=((\sigma_{ij}))$, where $\sigma_{i,i}=1$ for $i=1,\cdots,K$,  $\sigma_{1,2}=\sigma_{2,1}=0.5,\ \sigma_{3,4}=\sigma_{4,3}=0.3,\ \sigma_{7,8}=\sigma_{8,7}=0.3$, and all other elements of $\Sigma$ are $0$.
As most of the components on the copula function are assumed to be independent and the dependency exists only pairwise, the conditional copulas is easily obtained using properties of the bivariate normal distribution. The inverse of such conditional distributions can also be obtained using the inverse of the univariate normal distribution. Hence, we  use Algorithm \ref{lhsdc} to sample from the inputs using LHSD. The same conditional copulas are  used to sample from the joint distribution using the SRS, those are also used in the first step of rank based LHS method.

We sample $30$ observations from the inputs using all these five methods and computed the corresponding value of the outputs using \eqref{rivereq}. The value of $\tau$ is estimated by the sample mean of the  outputs. We repeated this procedure $10,000$ times to estimate the bias and variance of the estimator using these five methods. We also repeated the experiment by taking sample sizes as $20$, $75$, and $100$. The results are shown in Table \ref{tab:river_output}. The empirical  probability distributions of the estimators are shown in Figure \ref{river_fig2}. Here, the target value of $\tau$ is computed by sampling $10,000$ inputs using the SRS method. Our LHSD-based methods give the best results  in terms of both bias and MSE of the estimator. It is interesting to observe that the variance of the estimator using the SRS is very high, and for that reason even the biased LHS(ind) gives a smaller MSE than the SRS.

\begin{table}[h]
    \centering
    \begin{tabular}{l|rrr|rrr}
\multirow{2}{*}{Methods}   &  Bias   &   Variance     &      MSE &  Bias   &   Variance     &      MSE \\
\cline{2-7}  &     \multicolumn{3}{c|}{$N=20$}   &   \multicolumn{3}{c}{$N=30$}  \\
         \hline
       LHSD  & 0.013 & 0.010 & 0.010 & 0.062 & 0.008 & 0.011\\
       $\rm{LHSD_c}$ & 0.010 & 0.005 & 0.005 & 0.061 & 0.005 & 0.009\\
   LHS(ind) &  0.240 & 0.018 & 0.076   & 0.339 & 0.014 & 0.129\\
 LHS(Stien) &  0.029 & 0.013 & 0.014  & 0.073 & 0.009 & 0.015\\
        SRS & 0.013 & 0.853 & 0.853  & 0.079 & 0.819 & 0.825\\
\hline
Methods &     \multicolumn{3}{c|}{$N=75$}   &   \multicolumn{3}{c}{$N=100$}  \\
\hline
LHSD  & 0.017 & 0.006 & 0.006  & 0.002 & 0.005 & 0.005\\
 $\rm{LHSD_c}$ & 0.017 & 0.005 & 0.005  & 0.002 & 0.005 & 0.005 \\ 
   LHS(ind) & 0.455 & 0.011 & 0.218 & 0.508 & 0.011 & 0.269\\ 
   LHS(Stien) & 0.028 & 0.006 & 0.007 & 0.010 & 0.006 & 0.006\\ 
   SRS & 0.030 & 0.860 & 0.861 & -0.006 & 0.822 & 0.822 \\ 
   \hline
     \end{tabular}
    \caption{The bias, variance, and MSE of $\sqrt{N}\hat{\tau}$ using five different sampling methods, where samples are drawn from the river flood inundation model.}
    \label{tab:river_output}
\end{table}

\begin{figure}
\centering
\includegraphics[height=8cm, width=12cm]{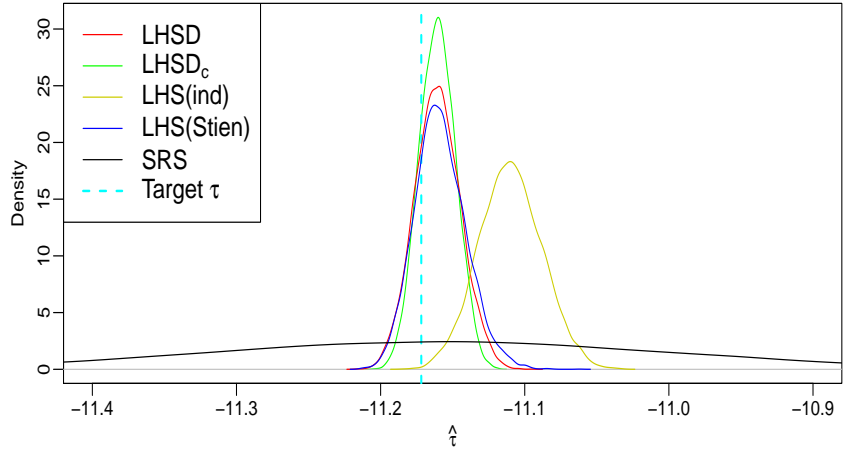}
\caption{The empirical density plot of $\hat{\tau}$ using five different sampling methods, where samples are drawn from the river flood inundation model.}
\label{river_fig2}
\end{figure}

The biases and MSE of the estimators for all pairwise correlation coefficients are also computed  using the five sampling methods (see Table \ref{tab:river_output2}). The true correlation coefficients for the first three pairs in that table, i.e. $(Q,K_s)$,  $(Z_v,Z_m)$, and $(L,B)$, are non-zero.  The last column for both bias and MSE combines all other pairs where the variables are pairwise independent, and we reported the maximum absolute bias and MSE for those pairs.  It shows that all methods except LHS(ind) estimate the correlation coefficient efficiently. The LHSD and LHSD$_c$ methods gives smaller bias and MSE for the  correlation coefficient in comparison to the ranked based LHS method. However, combining other results, we observe that the centered LHSD produces the best samples to estimate the output mean.


\begin{table}[h]
    \centering
    \setlength\tabcolsep{1.5pt} 
    \begin{tabular}{l|cccc|cccc}
\multirow{2}{*}{Methods}    &     \multicolumn{4}{c|}{Bias}   &   \multicolumn{4}{c}{MSE}   \\
\cline{2-9} & $(Q,K_s)$ & $(Z_v,Z_m)$ & $(L,B)$ & $\max(|\rm{Other}|)$ & $(Q,K_s)$ & $(Z_v,Z_m)$ & $(L,B)$ & $\max(|\rm{Other}|)$\\
         \hline
LHSD        & -0.026 & -0.006 & -0.009  &   0.003  & 0.019 & 0.027 & 0.027 &  0.036\\
$\rm{LHSD_c}$    & -0.023 & -0.009 & -0.011  &   0.003  & 0.018 & 0.027 & 0.027  &  0.036\\ 
LHS(ind)    & -0.502 & -0.297 & -0.300  &   0.003  & 0.286 & 0.123 & 0.125  &  0.036\\
LHS(Stien)  & -0.038 & -0.015 & -0.013  &   0.002  & 0.026 & 0.030 & 0.030 &  0.036\\
SRS         & -0.017 & -0.005 & -0.004  &   0.002  & 0.023 & 0.029 & 0.029 &  0.036\\
\hline
     \end{tabular}
    \caption{The biases and MSEs of the correlation coefficients for different pairs of variables, where samples are drawn from the river flood inundation model. }
    \label{tab:river_output2}
\end{table}

\section{Conclusion}
In this paper, we have presented some new sampling methods, based on Latin hypercube sampling, to estimate the mean of the output variable when the inputs are treated as random variables with dependent components. It is theoretically proved that the asymptotic variance of the estimator is smaller than that of the estimator based on the SRS. The simulation results show that our estimators outperform the SRS, and also give smaller mean squared errors than the existing LHS methods. The most advantage of our method is that it retains the exact joint distribution of the input variables. To the best of our knowledge, there does not exist any other sampling scheme based on the LHS that satisfies this property in case of dependent inputs.  

\bibliography{LHSD_Reference}

\section*{Appendix}
\begin{proof}[Proof of Theorem  \ref{theorem:var_lhsd}.]
Suppose  $\mathbf{Z}_1, \mathbf{Z}_2, \cdots, \mathbf{Z}_N$ form a LHS from $\mathbf{Z}$,  where the components of $\mathbf{Z}$ are i.i.d.  $U(0,1)$ variables. Then, using the inverse transformation in Equation (\ref{inv_transformation}) we get
$$\hat{\tau} = \frac{1}{N} \sum_{i=1}^N h({\mathbf{X}_j}) =  \frac{1}{N} \sum_{i=1}^N h^*(\mathbf{Z}_j),$$
where  $\mathbf{X}_1, \mathbf{X}_2, \cdots, \mathbf{X}_N$ form a LHSD  from $\mathbf{X}$ having a target distribution $F$. For $0\leq y_1, y_2 < 1$, we define
$$ G_N(y_1, y_2) =\left\{
\begin{array}{ll}
1, & \mbox{ if } [Ny_1] = [Ny_2]\\
0, & \mbox{ otherwise,}
\end{array}
\right.
$$
where $[y]$ denotes the greatest integer less than or equal to $y$. Following \cite{mckay1},  the joint density of $\mathbf{Z}_1$ and $\mathbf{Z}_2$ can be written as
$$f(\mathbf{z}_1, \mathbf{z}_2) = \left(\frac{N}{N-1}\right)^K \prod_{k=1}^K (1 - G_N(z_{1k}, z_{2k})),$$
where $\mathbf{z}_j = (z_{j1}, z_{j2}, \cdots, z_{jK})^T$ for $j=1,2$. If $A^K = \{{\mathbf{z}} \in \mathbb{R}^K : ||{\mathbf{z}}|| \leq 1 \}$, then
\begin{align*}
    \rm{Cov}(h^*(\mathbf{Z}_1), h^*(\mathbf{Z}_2)) = \left(\frac{N}{N-1}\right)^K  \int_{A^K} \int_{A^K} h^*(\mathbf{z}_1) h^*(\mathbf{z}_2) \prod_{k=1}^K (1 - G_N(z_{1k}, z_{2k})) d\mathbf{z}_1 d\mathbf{z}_2 - \tau^2.
   \end{align*}
Using $\left(\frac{N}{N-1}\right)^K = 1 + \frac{K}{N} + O\left(\frac{1}{N^2}\right)$ and keeping only the first order terms in the expansion of the product, we get 
\begin{align}
      &\rm{Cov}(h^*(\mathbf{Z}_1), h^*(\mathbf{Z}_2))\nonumber\\
      & =  \left(\frac{N}{N-1}\right)^K  \int_{A^K} \int_{A^K} h^*(\mathbf{z}_1), h^*(\mathbf{z}_2)  d\mathbf{z}_1 d\mathbf{z}_2 \nonumber\\
     & \ \ \ \ -  \left(\frac{N}{N-1}\right)^K \sum_{k=1}^K \int_{A^K} h^*(\mathbf{z}_1) h^*(\mathbf{z}_2)  G_N(z_{1k}, z_{2k}) d\mathbf{z}_1 d\mathbf{z}_2 - \tau^2 + O\left(\frac{1}{N^2}\right) \nonumber\\
      &=  \left(\frac{K}{N} + O\left(\frac{1}{N^2}\right)\right) \tau^2 - \left(1 +O\left(\frac{1}{N}\right)\right) \sum_{k=1}^K \int_{A^K} \int_{A^K} h^*(\mathbf{z}_1), h^*(\mathbf{z}_2)  G_N(z_{1k}, z_{2k}) d\mathbf{z}_1 d\mathbf{z}_2 \nonumber\\ 
      & \ \ \ \ + O\left(\frac{1}{N^2}\right) .
     \label{cov4}
    \end{align}
Define $I_j = [(j-1)/N, j/N]$. Using Equation (\ref{def:alpha}) we have
\begin{align}
      \int_{A^K} \int_{A^K} & h^*(\mathbf{z}_1), h^*(\mathbf{z}_2)  G_N(z_{1k}, z_{2k}) d\mathbf{z}_1 d\mathbf{z}_2  
      \nonumber\\ & = \int_{A^1} \int_{A^1} (\alpha_k(z_{1k}) - \tau ) (\alpha_k(z_{2k}) - \tau)  G_N(z_{1k}, z_{2k}) dz_{1k} dz_{2k} \nonumber\\
      & = \sum_{j=1}^N \left( \int_{I_j}  (\alpha_k(z_k) - \tau )  dz_k  \right)^2 .
      \label{int1}
    \end{align}
If $\alpha_k^2$ is Riemann integrable, one can show that
\begin{equation}
     N \sum_{j=1}^N \left( \int_{I_j}  (\alpha_k(z_k) - \tau )  dz_k  \right)^2 \longrightarrow \int_0^1  (\alpha_k(z_k) - \tau )^2  dz_k. 
\label{int}
\end{equation}
Combining Equations (\ref{cov4}), (\ref{int1}) and (\ref{int})  we have
\begin{align}
      \rm{Cov}(h^*(\mathbf{Z}_1), h^*(\mathbf{Z}_2)) =  \frac{K}{N} \tau^2 -  \frac{1}{N}\sum_{k=1}^K \int_0^1  (\alpha_k(z_k) - \tau )^2 dz_k + o\left(\frac{1}{N}\right) .
     \label{cov3}
    \end{align}
Therefore
\begin{align*}
{\rm{Var_{LHSD}}}(\hat{\tau}) & = \frac{1}{N} \rm{Var}(h^*(\mathbf{Z}_1)) +  \frac{N-1}{N}  \rm{Cov}(h^*(\mathbf{Z}_1), h^*(\mathbf{Z}_2)) \nonumber\\
& = \frac{1}{N} \left\{ \int_{A^K} h^{*2}(\mathbf{z}) d\mathbf{z} -\tau^2 + K \tau^2 -  \sum_{k=1}^K \int_0^1  (\alpha_k(z_k) - \tau )^2 dz_k \right\} + o\left(\frac{1}{N}\right)\nonumber\\
& = \frac{1}{N} \left\{ \int_{A^K} h^{*2}(\mathbf{z}) d\mathbf{z}  - \tau^2 -  \sum_{k=1}^K \int_0^1  \alpha_k^2(z_k)  dz_k \right\} + o\left(\frac{1}{N}\right)\nonumber\\
& = \frac{1}{N}  \int_{A^K} \left\{ h^*(\mathbf{z})   - \tau -  \sum_{k=1}^K   \alpha_k(z_k)   \right\}^2 d\mathbf{z} + o\left(\frac{1}{N}\right)\nonumber\\
& = \frac{1}{N} \int_{A^K} r^2({\mathbf{z}}) d{\mathbf{z}} + o\left(\frac{1}{N}\right).
\label{var_LHS1}
\end{align*}
\end{proof}

\begin{proof}[Proof of Theorem  \ref{theorem:asymp_dist}.]
From Theorem \ref{theorem:var_lhsd} we find that the asymptotic variance of $N^{1/2} (\hat{\tau} - \tau )$ is $\int_{A^K} r^2({\mathbf{z}}) d{\mathbf{z}}$. As $E(\hat{\tau}) = \tau$, it is enough to prove that the asymptotic distribution of $N^{1/2} (\hat{\tau} - \tau )$ is normal. We write
\begin{equation}
    N^{1/2} (\hat{\tau} - \tau ) = N^{1/2} \left(\sum_{k=1}^K \bar{\alpha}_k + \bar{R} \right),
\end{equation}
where $$\bar{\alpha}_k = \frac{1}{N} \sum_{j=1}^N \alpha_k (Z_{jk}) \mbox{ and } \bar{R} = \frac{1}{N} \sum_{j=1}^N r(\mathbf{Z}_i).
$$
Theorem \ref{theorem:var_lhsd} shows that ${\rm{Var_{LHSD}}}(N^{1/2}\bar{\alpha}_k) = o(1/N)$. Since $E(\bar{\alpha}_k) = \int_0^1 \alpha_k(z_k) dz_k = 0$, it shows that $N^{1/2}\bar{\alpha}_k = o_p(1/N)$. Finally, the theorem follows from Lemma 2 of \cite{Owen} that proves  $N^{1/2}\bar{R}$ is asymptotically normal under the LHS.
\end{proof}
\end{document}